\documentclass[]{article}

\newcommand*{\myTitle}{A Communication-Efficient Distributed Data Structure for Top-$k$ and $k$-Select Queries}

\usepackage[utf8]{inputenc}

\usepackage{amsmath}
\usepackage{amssymb}
\usepackage{amsthm}
\usepackage{graphicx}
\usepackage{bm}
\usepackage{csquotes}
\usepackage{enumitem}
\usepackage[colorlinks=true]{hyperref}
\usepackage{listings}
\usepackage{mathtools}
\usepackage{nicefrac}
\usepackage{todonotes}
\usepackage{xargs}
\usepackage{xfrac}
\usepackage{xspace}
\usepackage{dsfont}
\usepackage{relsize}

\usepackage{algpseudocode}
\usepackage{algorithm}

\usepackage{array}

\usepackage{url}
\urldef{\mailsa}\path|{felixbm, bjoernf, malatya, fmadh}@mail.upb.de|

\usepackage[capitalise,nameinlink,noabbrev]{cleveref}
\crefname{step}{Step}{Steps}
\creflabelformat{step}{#2#1#3}

\usepackage{calrsfs}
\DeclareMathAlphabet{\pazocal}{OMS}{zplm}{m}{n}

\usepackage{enumitem}

\usepackage{marginnote}

\newcommand*{\myThanks}{This work was partially supported by the German Research Foundation (DFG) within the Priority Program ``Algorithms for Big Data'' (SPP 1736).}

\title{\myTitle\thanks{\myThanks}}

\author{Felix Biermeier \and Björn Feldkord \and Manuel Malatyali \and
	Friedhelm Meyer auf der Heide \\ [0.4em]
	Heinz Nixdorf Institute \& 	Computer Science Department\\
	Paderborn University, Germany\\[0.2em]
	\{felixbm, bjoernf, malatya, fmadh\}@mail.upb.de}
\date{}

\newcommand{\bigO}{\pazocal{O}}

\newcommand*{\polylog}{\text{polylog}\, }

\newcommand*{\ALG}{\texttt{ALG}\xspace}

\newcommand*{\E}{\mathbb{E}}

\newcommand*{\squash}{\vspace*{-0.25cm}}

\newcommand{\topkproblem}{Top-$k$-Problem\xspace}
\newcommand{\topk}{Top-$k$\xspace}
\newcommand{\topkprotocol}{Top-$k$ Protocol\xspace}
\newcommand*{\kselect}{$k$-Select\xspace}
\newcommand*{\aproxKSelectProblem}{approx.\ $k$-Select Problem\xspace}
\newcommand*{\aproxkSelectProtocol}{Approx.\ $k$-Select Protocol\xspace}

\newcommand{\CoFaSel}{\textsc{CoFaSel}\xspace}
\newcommand{\CoFaSelAmp}{\textsc{CoFaSelAmp}\xspace}
\newcommand*{\sample}{\mathcal{S}_{\varepsilon,\delta}}

\newcommand*{\h}{\mathcal{H}}
\newcommand{\con}{\kappa}

\newcommand*{\logphilog}{\log_{\nicefrac{1}{\phi}} (\log (n))}

\newcommand*{\al}{a_\ell}
\newcommand*{\rl}{r_\ell}

\newcommand{\comment}[2]{}

\newcommand{\man}[1]{\comment{MM}{#1}}

\newtheorem{theorem}{Theorem}[section]
\newtheorem{lemma}[theorem]{Lemma}
\newtheorem{corollary}[theorem]{Corollary}
\newtheorem{definition}[theorem]{Definition}

\begin{document}
\maketitle

\begin{abstract}
We consider the scenario of $n$ sensor nodes observing  streams of data.
The nodes are connected to a central server whose task it is to compute some function over all data items observed by the nodes.
In our case, there exists a total order on the data items observed by the nodes.
Our goal is to compute the $k$ currently lowest observed values or a value with rank in $[(1-\varepsilon)k,(1+\varepsilon)k]$ with probability $(1-\delta)$.
We propose solutions for these problems in an extension of the distributed monitoring model where the server can send broadcast messages to all nodes for unit cost.
We want to minimize communication over multiple time steps where there are $m$ updates to a node's value in between queries.
The result is composed of two main parts, which each may be of independent interest:
\begin{enumerate}
  \item Protocols which answer \topk and \kselect queries. These protocols are memoryless in the sense that they gather all information at the time of the request.
	\item A dynamic data structure which tracks for every $k$ an element close to $k$.
\end{enumerate}
We describe how to combine the two parts to receive a protocol answering the stated queries over multiple time steps.
Overall, for Top-$k$ queries we use $\bigO(k + \log m + \log \log n)$ and for $k$-Select queries $\bigO(\frac{1}{\varepsilon^2} \log \frac{1}{\delta} + \log m + \log^2 \log n)$ messages in expectation.
These results are shown to be asymptotically tight if $m$ is not too small.
\end{abstract}

\newpage
\section{Introduction}

Consider a distributed sensor network which is a system consisting of a huge amount of nodes.
Each node continuously observes its environment and measures information (e.g.\ temperature, pollution or similar parameters).
We are interested in aggregations describing the current observations at a central server.
To keep the server's information up to date, the server and the nodes can communicate with each other.
In sensor networks, however, the amount of such communication is particularly crucial, as communication has the largest impact to energy consumption, which is limited due to battery capacities \cite{battery}.
Therefore, algorithms aim at minimizing the (total) communication required for computing the respective aggregation function at the server.

We consider several ideas to potentially lower the communication used.
Each single computation of an aggregate should use as little communication as possible.
Computations of the same aggregate should \emph{reuse} parts of previous computations.
Only compute aggregates, if necessary.
Recall that the continuous monitoring model creates a new output as often as possible.

\subsection{Model}
We consider the distributed monitoring model, introduced by Cormode, Muthukrishnan, and Yi in \cite{cormodeSurvey}, in which there are $n$ distributed nodes, each uniquely identified by an identifier (ID) from the set $ \{1,\dots,n \}$, connected to a single server.
Each node observes a stream of data items over time, i.e.\ at any discrete time step $t$ node $i$ observes a data item $d_i^t$.
We assume that the data items have a total order and denote by \emph{rank($d$)} the position of data item $d$ in the sorted ordering.
Furthermore, we assume that the sorted ordering is unique, i.e.\ for each data items $d_i$ and $d_j$ either $d_i \leq d_j$ or $d_i \geq d_j$ holds.

The server is asked to, given a query at time $t$, compute an output $f(t)$ which depends on the data items $d_i^{t}$ with $i = 1, \ldots, n$ observed across all distributed streams.
The exact definition of $f(\cdot)$ depends on the concrete problems under consideration, which are defined in the section below.
For the solution of these problems, we are interested in both, exact and approximation algorithms.
An exact algorithm computes the (unique) output $f(t)$ with probability $1$. 
An $\varepsilon$-approximation of $f(t)$ is an output $\tilde f(t)$ of the server such that $(1-\varepsilon)f(t) \leq \tilde f(t) \leq (1+\varepsilon)f(t)$ holds.
We call an algorithm that provides an $\varepsilon$-approximation with probability at least $1-\delta$, an $(\varepsilon, \delta)$-approximation algorithm.
We say an algorithm is correct \emph{with high probability}, if for a given constant $c > 1$ it is correct with probability at least $1-n^{-c}$.

\paragraph{Communication Network} 
To be able to compute the output, the nodes and the server can communicate with each other by exchanging single cast messages or by broadcast messages sent by the server and received by all nodes.
Both types of communication are instantaneous and have unit cost per message.
That is, sending a single message to one specific node incurs cost of one and so does one broadcast message.
Each message has a size of $\bigO \left(\mathcal{B} + \log (n) + \log( \log (\frac{1}{\delta}))\right)$ bits, where $\mathcal{B}$ denotes the number of bits needed to encode a data item.
A message will usually, besides a constant number of control bits, consist of a data item, a node ID and an identifier to distinguish between messages of different instances of an algorithm applied in parallel (as done when using standard probability amplification techniques).
A broadcast channel is an extension to \cite{cormodeSurvey}, which was originally proposed in \cite{cormodeFunction} and afterwards applied in \cite{mmm15,mmm16,sirocco}.
Between any two time steps we allow a communication protocol to take place, which may use a polylogarithmic number of rounds.
The optimization goal is the minimization of the communication complexity, given by the total number of exchanged messages, required to answer the posed requests or rebuild the data structure.

\paragraph{Problem Description}
In this work, we consider two basic problems that are related to the rank of the data items: 
(1) Compute exactly (all of) the $k$ smallest data items observed by the nodes at the current time step $t$ and
(2) output an $(\varepsilon, \delta)$-approximation of the data item with rank $k$.

Formally, let $\pi_t$ denote the permutation of the node IDs $\{1, \ldots, n\}$ such that $\pi_t(i)$ gives the index of the data item with rank $i$ at time $t$, 
i.e. $i = rank(d_{\pi_t(i)}^t)$.
First, we denote the \topkproblem as the output $\{d_{\pi_t(1)}^t, \ldots, d_{\pi_t(k)}^t\}$ for a given $1 \leq k \leq n$ and we consider exact algorithms for this problem.
Second, we consider the \aproxKSelectProblem which is to output one data item $d \in \{d_{\pi_t((1-\varepsilon) k)}^t, \ldots, d_{\pi_t((1+\varepsilon) k)}^t \}$. 
We consider $(\varepsilon, \delta)$-approximation algorithms, i.e.\ an algorithm outputs such a data item $d$ correctly with probability at least $1-\delta$.

\paragraph{Distributed Data Structure}
We develop a data structure which supports the following operations:

\begin{enumerate}[leftmargin=3cm]	
	\item[UPDATE$(i,d)$:] Node $i$ receives a new data item $d$.
	\item[INITIALIZE$()$:] Set up the data structure (the nodes may already have observed data items).
	\item[REFRESH$()$:] Is called by the server if a request is made and the data structure is already initialized.
	\item[ROUGH-RANK$(k)$:] Returns a data item $d$ where $rank(d)\in[k, k\cdot \log^{c}(n)]$ holds with probability at least $1-\log^{c'}(n)$ for some suitable constants $c, c'$.
\end{enumerate}

ROUGH-RANK queries are used to receive an element as a basis for further computations in our protocols.
We consider a distributed data stream in our model as a sensor node observing its environment by explicitly calling an UPDATE operation to overwrite its previous observation by a new one.

\subsection{Our Contribution}
In this paper we propose exact algorithms for the $k$ currently lowest observed values or approximation algorithms for a value with rank in $[(1-\varepsilon)k,(1+\varepsilon)k]$ with probability $(1-\delta)$.
 Our data structure is based on single-shot computations which are of independent interest:
\begin{table}[ht]
\begin{center}
\begin{tabular}[ht]{
		@{\hspace{0.2cm}} r @{\hspace{0.2cm}} | 
		@{\hspace{0.2cm}} c @{\hspace{0.2cm}} | 
		@{\hspace{0.2cm}} c @{\hspace{0.2cm}} | 
		@{\hspace{0.2cm}} c @{\hspace{0.2cm}} | 
		@{\hspace{0.2cm}} c @{\hspace{0.2cm}} | 
		@{\hspace{0.2cm}} c @{\hspace{0.2cm}}}
	 
	problem & type & bound & (total) comm. &  comm.\ rounds & Reference\\
	\hline
	Top-$k$ & exact & upper & $\bigO(k \cdot \log (n))$ & $\bigO(k \cdot \log(n))$ & \cite{mmm15} \\
	Top-$k$ & exact & upper & $k + \log(n) + 1$ & $\bigO(k + \log(n))$ & \cref{sec:topk}\\
	Top-$k$ & exact & upper & $\bigO(k + \log(n))$ & $\bigO(\log(\frac{n}{k}))$ & \cref{se:aproxkSelect} \\
	Top-$k$ & exact & lower & $k + \Omega(\log(n)) $ & / & \cref{se:lowerBounds} \\
	$k$-Select & approx. & upper & $\bigO(\frac{1}{\varepsilon^2} \log (\frac{1}{\delta}) + \log (n))$ & $\bigO(\log(\frac{n}{k}))$ & \cref{se:aproxkSelect} \\
	$k$-Select & approx. & lower & $\Omega(\log (n))$ & / & \cref{se:lowerBounds}
\end{tabular}
\caption{Summary of results for single-shot computations.}
\end{center}
\end{table}

\vspace*{-0.8cm}
\noindent With notion of the data structure this relates to the results as presented in \cref{tb:ds}.
\begin{table}[ht]
\begin{center}
\begin{tabular}[ht]{
	@{\hspace{0.2cm}} r @{\hspace{0.2cm}} | 
	@{\hspace{0.2cm}} c @{\hspace{0.2cm}} | 
	@{\hspace{0.2cm}} c @{\hspace{0.2cm}} | 
	@{\hspace{0.2cm}} c @{\hspace{0.2cm}}} 

	operation & (total) comm. &  comm.\ rounds & Reference\\
	\hline
	INITIALIZE & $\bigO(\log(n))$, $\Omega(\log(n))$ & $\bigO(\log(n))$ & \cref{se:init}\\
	REFRESH & $\bigO(\log(m))$, $\Omega(\log(m))$ & $\bigO(\log(n))$ & \cref{se:refresh} \\
	UPDATE & amortized & 1 & \cref{sec:DS} \\
	ROUGH-RANK & 0 & 0 & \cref{sec:DS} \\
\end{tabular}
\caption{Summary of results using our new data structure.}
\label{tb:ds}
\end{center}
\end{table}

In Section~\ref{sec:final} we describe how to combine the one-shot computations with the given data item from the ROUGH-RANK operation supported by our data structure.
This leads to the overall bound of $\Theta(k + \log(m) + \log(\log(n)))$ messages 
in expectation for \topk queries. 
The bound on the number of messages is asymptotically tight if 
$m \geq \log(n)$ holds where $m$ denotes the number of UPDATEs since the last query. 
For \kselect queries and applying the same combination as above, this protocol uses $\Theta\! \left(\frac{1}{\varepsilon^2} \log (\frac{1}{\delta}) + \log(m) + \log^2(\log(n)) \right)$ messages in expectation, where the bound on the number of messages is asymptotically tight if $m \geq \log^{\log(\log(n))}(n)$ holds.

Furthermore, we parameterize our algorithms such that it is possible to choose a trade-off between the number of messages used and the number of communication rounds, i.e.\ time, used to process the queries.

\subsection{Related Work}
\label{sec:relatedWork}
Cormode, Muthukrishnan, and Yi introduce the Continuous Monitoring Model \cite{cormodeSurvey} with an emphasis on systems consisting of $n$ nodes generating or observing \emph{distributed} data streams and a designated coordinator. 
In this model the coordinator is asked to continuously compute a function, i.e.\ to compute a new output with respect to all observations made up to that point.
The objective is to aim at minimising the total communication between the nodes and the coordinator.
We enhance the continuous monitoring model (as proposed by Cormode, Muthukrishnan, and Yi in \cite{cormodeFunction}) by a broadcast channel. 
Note, that we are not strictly continuous in the sense that we introduce a dynamic data structure which only computes a function, if there is a query for it. 
However, there is still a continuous aspect: 
In every time step, our data structure maintains elements close to all possible ranks in order to quickly answer queries if needed.

An interesting area of problems within this model are threshold functions:
The coordinator has to decide whether the function value (based on all observations) has reached a given threshold $\tau$. 
For well structured functions (e.g.\ count-distinct or the sum-problem) asymptotically optimal bounds are known \cite{cormodeSurvey,cormodeFunction}.
Functions which do not provide such structures (e.g.\ the entropy \cite{chakrabarti}), turn out to require much more communication volume.

A related problem is considered in \cite{babcock}.
In their work, Babcock and Olston consider a variant of the distributed top-$k$ monitoring problem:
There is a set of objects $\{O_1, \ldots, O_n\}$ given, in which each object has a numeric value. 
The stream of data items updates these numeric values (of the given objects). 
In case each object is associated with exactly one node, their problem is to monitor the $k$ largest values.
Babcock and Olston have shown by an empirical evaluation, that the amount of communication is by an order of magnitude lower than that of a naive approach.

A model related to our (sub-)problem of finding the $k$-th largest values, and exploiting a broadcast channel is investigated by the shout-echo model \cite{marberg,rotem}:
A communication round is defined as a broadcast by a single node, which is replied by all remaining nodes. 
The objective is to minimise the number of communication rounds, which differs from ours.

\section{One shot computation: Top-K}
\label{sec:topk}
In this section we present an algorithm which identifies all data items with rank at most $k$ currently observed by the sensor nodes.
Note that when we later apply this protocol for multiple time steps, not all sensor nodes might participate.
In this section, we denote by $N$ the number of participating nodes.

Our protocol builds a simple search-tree-like structure based on a height the nodes draw from a geometric distribution.
Afterwards, a simple strategy comparable to an in-order tree walk is applied.
In order to identify the smallest data item this idea is implemented as follows:
The protocol draws a uniform sample of (expected) size $\frac{1}{\phi}$ and broadcasts the smallest data item.
Successively, the protocol chooses a uniform sample until the smallest item is identified. 
In this description, each drawing of a sample corresponds to consider all children of the current root of
the search tree and then continue with the left-most child as the new root.

The protocol is given a maximal height $h_{max}$ for the search tree which corresponds to the number of repetitions of the protocol described above.
We define a specific value for $h_{max}$ in \cref{th:topk}.
Furthermore, the algorithm is given a parameter $\phi$ which defines the failure probability of the geometric distribution. 
\begin{algorithm}[ht]
		\begin{minipage}[t]{0.45\linewidth-1em}
			Initialization() \squash
			\begin{enumerate}
				\item Each node $i$ draws a random variable $h_i$, i.i.d.\ from a geometric distribution with $p = 1-\phi$ 
				\item Server defines $\ell \coloneqq -\infty$, $u \coloneqq \infty$, $h \coloneqq h_{max}$ and $S \leftarrow \emptyset$
				\item \textbf{Call} Top-$k$-Rec$(\ell, u, h)$
				\item Raise an \textbf{error}, if $|S| < k$
			\end{enumerate}\squash
		\end{minipage}
		~~~~~~~~~~
		\begin{minipage}[t]{0.53\linewidth-1em}
			Top-$k$-Rec($\ell, u, h)$ \squash
			\begin{enumerate}[noitemsep]
				\item \textbf{If} $h = 0$ \textbf{then} 
				\item \quad \textbf{if} $|S| = k$ \textbf{then} return $S$,
				\item \quad \textbf{Else} end recursion
				\item Server probes sensor nodes $i$ with \\ $\ell < d_i < u$ and $h_i \geq h$ \\
				 Let $r_1< \ldots< r_j$ be the responses
				\item \textbf{If} there was no response \textbf{then}
				\item \quad \textbf{Call} Top-$k$-Rec$(\ell, u, h-1)$
				\item \textbf{Else}
				\item \quad\textbf{Call} Top-$k$-Rec$(\ell, r_1, h-1)$
				\item \quad$S \leftarrow S \cup r_1$
				\item \quad\textbf{For} $i = 1$ to $j-1$ do 
					\item \hspace*{0.5cm} \textbf{Call} Top-$k$-Rec$(r_i, r_{i+1}, h-1)$
					\item \hspace*{0.5cm} $S \leftarrow S \cup r_{i+1}$
				\item \quad\textbf{Call} Top-$k$-Rec$(r_j, u, h-1)$
			\end{enumerate}
		\end{minipage}
	\caption{\topkprotocol($\phi, h_{max}$)}
	\label{alg:topk}
\end{algorithm}

The algorithm starts by drawing a random variable $h_i$ from a geometric distribution, i.e.\ $\Pr[h_i = h] = \phi^{h-1} (1-\phi)$.
We discuss the choice of $\phi$ at the end of this section. 
Note that $\phi$ enables a trade off between the number of messages sent in expectation and communication rounds used.

The protocol can be implemented in our distributed setting by having the server broadcast the values
$\ell, u,$ and $h$ such that each node with the corresponding height values and data items responds.
Note, the variables $r_1, \ldots, r_j$ used in Steps $8, 9, 11, 12,$ and $13$ refer to responses of the current call of Top-$k$-Rec.

\subsubsection*{Analysis}
In the following we show that the expected number of messages used by the \topkprotocol is upper bounded by $k + \frac{1-\phi}{\phi}\cdot\log_{\nicefrac{1}{\phi}} (N) + 1$ in \cref{th:topk}. 
Afterwards, an upper bound of $\bigO(\phi \cdot k + h_{max})$ on the number of communication rounds is presented in \cref{le:topk-rounds}.
Defining $\phi \coloneqq 1/2$ the bound on the communication translates to a tight bound of $k + \log (n)+1$ in \cref{co:topk-opt} complemented by a simple lower bound of $k + \Omega(\log (n))$ in \cref{se:lowerBounds}.

We show an upper bound on the communication used by the \topkprotocol analyzing the expected value of a mixed distribution. 
The analysis works as follows:
We sort the nodes by their rank and determine the number of nodes with height $\leq h$ before the first node with a height $>h$ in this ordering by a geometric-sequence in \cref{def:geoSeq}.
For each height $h$, this number can then be used to determine the number of nodes which send a message on height $h$, which we model by a geocoin-experiment in \cref{def:geoExp}.
Note that this analysis turns out to be very simple since independence can be exploited in a restricted way and leads to a proper analysis with respect to exact constants. 

\begin{definition}
	\label{def:geoSeq}
	We call a sequence $G = (G_1, \ldots, G_m)$ of $m$ random experiments a \emph{geometric-sequence}, if each $G_h$ is chosen from a geometric distribution with $p_{h}^{geo} \coloneqq \phi^h$.
	We denote its $size(G) \coloneqq \sum_h G_h $ and say it \emph{covers} all nodes, if $size(G) \geq N$. 
\end{definition}

For the analysis, we choose a fixed length of $m \coloneqq \log_{\nicefrac{1}{\phi}} (N)$ and modify $G$ to $G' = (G_1, \ldots, G_{m-1}, N)$
such that $G'$ covers all nodes with probability 1.

Based on a given geometric-sequence, we define a sequence describing the number of messages send by the nodes on a given height.
We take the number of nodes $G_i$ as a basis for a Bernoulli experiment where the success probability is the probability a node
sends a message on height $i$. This is $\Pr[h=h_i~|~h\leq h_i]=\frac{\phi^{h-1}(1-\phi)}{1-\phi^h}$.

\begin{definition}
	\label{def:geoExp}
	We denote a \emph{geocoin-experiment} by $C = (C_1, \ldots, C_m)$ of random variables $C_h$ which are drawn from $Binom (n=G_h, p_h^{bin} = \frac{\phi^{h-1}(1-\phi)}{1-\phi^h})$, i.e.\ $C_h$ out of $G_h$ successful coin tosses where each coin toss is successful with probability $p_h^{bin}$.
\end{definition}

\begin{theorem}
	\label{th:topk}
	Let $N>k$ and $h_{max} \geq \log_{\nicefrac{1}{\phi}} (N)$ hold. 
	The \topkprotocol uses at most $k + \frac{1-\phi}{\phi} \log_{\nicefrac{1}{\phi}} (N) + 1$ messages in expectation.
\end{theorem}

\begin{proof}
	The probability to send a message of a node $v$ within the \topk is $1$.
	It remains to show that the overhead is bounded by $\frac{1-\phi}{\phi}\ \log_{\nicefrac{1}{\phi}} (N) + 1$.

	The number of messages sent by \cref{alg:topk} (excluding the $k$ nodes observing the $k$ smallest data items) is upper bounded by a geocoin-experiment $C$. 
	Let $\pazocal{H} \coloneqq \log_{\nicefrac{1}{\phi}} (N)$.   
	For $h < \pazocal{H}$ we use that the geometric distribution is memory-less and hence 
	$ \E[C_h] = (1-p_h^{geo}) \cdot (p_h^{bin} + \E[C_h]) 
	           = (1- \phi^h) \cdot \left(\frac{\phi^{h-1}(1-\phi)}{1-\phi^h} + \E[C_i]\right). $
	This can simply be rewritten as $\E[C_i] = (1-\phi)/ \phi$.

  For $i \geq \pazocal{H}$ 
   we bound the number of messages by the total number of nodes with height at least $\pazocal{H}$.
	These can be described as the expectation of a Bernoulli experiment with $N$ nodes and success probability $\phi^{\pazocal{H}-1}$ and hence
	$\E[C_{\geq \pazocal{H}}]\leq\phi^{\pazocal{H}-1}\cdot N = 1/\phi$.

\noindent
  In total, we get
	$\sum_h \E[C_h] = 
	\left(\sum_{i=h}^{\pazocal{H}-1}\E[C_i] \right)+\E[C_{\geq \pazocal{H}}]
	\leq\frac{1-\phi}{\phi} \ \log_{\nicefrac{1}{\phi}}(N)+1.$
\end{proof}

\begin{lemma}
	\label{le:topk-rounds}
	The \topkprotocol needs $\bigO(\phi \cdot k + h_{max})$ communication rounds in expectation.
\end{lemma}
\begin{proof}
	We structure the proof in two steps: 
	First, we analyse the number of rounds used to determine the minimum, and second, the number of communication rounds used to determine the \topk.
	
	Observe, that the algorithm uses a linear amount of steps (linear in $h_{max}$), until it reaches $h = 1$, after which the minimum is found. 
	Afterwards in each step the algorithm recursively probes for nodes successively larger than the currently largest values, that are added to the output set $S$.
	Note, that by the analysis in Theorem~\ref{th:topk}, the number of nodes that send a message in expectation in each round is $(1-\phi)/\phi$ (for $h < \log_{\nicefrac{1}{\phi}}(N)$).
	Thus, in each communication round there are $\Omega(\frac{1}{\phi})$ nodes in expectation that send a message, such that after $\bigO(\phi \cdot k)$ rounds in expectation the \topkprotocol terminates.
\end{proof}

Note that our bounds describe a trade off between the number of messages and communication rounds, where the number of messages decreases with a small success probability $1-\phi$. 
Intuitively speaking, this stems from more larger resulting height values such that the search structure has a smaller breadth.

\begin{corollary}
	\label{co:topk-opt}
	For $N = n$, $\phi \coloneqq \frac{1}{2}$, and $h_{max} \coloneqq \log (n)$, the \topkprotocol uses an amount of  $k+\log (n) + 1$ number of messages in expectation and $\bigO(k + \log (n))$ communication rounds.
\end{corollary}

\section{One Shot Computation: Approximate $k$-Select}
\label{sec:kselect}
In this section we present an algorithm which gives an $(\varepsilon, \delta)$-approximation for the $k$-Select Problem, i.e.\ a data item $d$ is identified with a rank between $(1-\varepsilon)k$ and $(1+\varepsilon)k$ with probability at least $1-\delta$. 

In \cref{se:CoFaSel}, we introduce an algorithm which identifies a data item with rank $\Theta(k)$.
This is done to reduce the number of messages for the algorithm proposed in \cref{se:aproxkSelect}
which uses a standard sampling technique to achieve the desired approximation.

\subsection{Constant Factor Approximation}
\label{se:CoFaSel}
The following algorithm employs a similar strategy as \cref{alg:topk}. 
However, the protocol terminates and outputs a data item as soon as the targeted height of $h_{min}$ is reached.
This data item is one of the responses on height $h_{min}$, dependent on the value of $\phi$.
Note that it may not be sufficient to output the smallest value, since the number of responses may be very large if $\phi$ is small.

\begin{algorithm}[ht]
	\begin{enumerate}
		\item Each node $i$ defines a random variable $h_i$,
		i.i.d.\ drawn from \\
		a geometric distribution with $p = (1-\phi)$, and redefines $h_i \coloneqq \min\{h_i, h_{max}\}$.
		\item Server defines $d_{min} \coloneqq \infty$. \hfill $\triangleright$ $\forall$ data items d: $\infty > d$
		\item Server defines $0 < \alpha < 1$, s.th.\ $\lfloor \log_{\nicefrac{1}{\phi}} (7k) \rfloor = \log_{\nicefrac{1}{\phi}} (7k) - \alpha \in \mathbb{N}$ holds.
		\item \textbf{for} $h \coloneqq h_{max}$ \textbf{to} $h_{min} = \log_{\nicefrac{1}{\phi}} ( 7 \, k )  - \alpha + 1$ \textbf{do} 
		\item \quad Server probes all nodes $i$ with $d_i < d_{min}$ and $h_i = h$.
		\item \quad Let $r_1 < r_2 < \ldots < r_j$ be the responses, ordered by their values.
		\item \quad \textbf{If} $h > h_{min}$ \textbf{then} Server redefines $d_{min} \coloneqq r_1$ \textbf{else} $d_{min} \coloneqq r_{(\nicefrac{1}{\phi})^\alpha}.$
		\item Output $d_{min}$
	\end{enumerate}
	\caption{\CoFaSel\!\!($h_{max},  \phi, k$) \hfill (ConstantFactorSelect)}
	\label{alg:CoFaSel}
	\vspace{-3mm}
\end{algorithm}

We show that \cref{alg:CoFaSel} outputs a data item with a rank larger than $k$ and smaller than $42 \ k$ with constant success probability in \cref{le:cofaselConst}. 
Furthermore, we state that \cref{alg:CoFaSel} determines a data item which is at most by a polylogarithmic factor larger than the expectation with high probability in \cref{le:cofaselBadError}.
We need this result later when reusing the protocol in Section~\ref{sec:DS}. 
An upper bound on the number of used messages is presented in \cref{le:cofaselApproxRounds}. 
We shortly state how to amplify the success probability to $1-\delta'$ for a given $\delta' > 0$ in \cref{th:CoFaSel}. 

\begin{lemma}
	\label{le:cofaselConst}
	The \CoFaSel Protocol outputs a data item $d$ with $rank(d)\in[k,42k]$ with probability at least $0.6$.
\end{lemma}
\begin{proof}
	The algorithm outputs the $(1/\phi)^\alpha$ smallest data item $d$ the server gets as a response on height $h = h_{min}$. 
	To analyze its rank simply consider the random number $X$ of nodes $i$ that observed smaller data items $d_i < d$. 
	The claim follows by simple calculations:
	(i) $\Pr[X < k] \leq \frac{1}{5}$ and (ii) $\Pr[42 k > X] \leq \frac{1}{5}$.
	
	The event that $X$ is (strictly) smaller than $k$ holds, if there are $(1/\phi)^\alpha$ out of $k$ nodes  with a random height at least $h_{min}$.  
	Let $X_1$ be drawn by a binomial distribution $Bin(n = k, p = \phi^{h_{min}-1})$. 
	It holds $\E[X_1] = k \cdot \phi^{h_{min}-1} = \frac{1}{7} \cdot (\frac{1}{\phi})^\alpha$.
	Then, $\Pr[X < k] \leq \Pr[ X_1 \geq (\frac{1}{\phi})^\alpha  ] = \Pr[X_1 \geq (1+6) \frac{1}{7 \phi^\alpha}] 
	\leq \exp(-\frac{1}{3} \frac{1}{7 \phi^\alpha} 6^2) \leq \frac{1}{5}$.
	
	On the other hand, the event that $X$ is (strictly) larger than $42k$ holds, if there are less than $(1/\phi)^\alpha$ out of $42k$ nodes with a random height of at least $h_{min}$. 
	Let $X_2$ be drawn by a binomial distribution $Bin(n = 42k, p = \phi^{h_{min}-1})$.
	It holds $\E[X_2] = (42 k) \phi^{h_{min}-1} = (42 k ) (7 k ) ^{-1} \phi ^{-\alpha} = \frac{6}{\phi^\alpha}$.
	Then, $\Pr[X > 42k] \leq \Pr[X_2 < \frac{1}{\phi^\alpha}]  = \Pr[X_2 < (1-(1-\frac{1}{6})) \frac{6}{\phi^\alpha}] $
	$\leq \exp(-\frac{1}{2} (\frac{6}{\phi^\alpha} (1-\frac{1}{6})^2) \leq \exp(-\frac{25}{12}) \leq \frac{1}{5}$.
\end{proof}

\begin{lemma}
	\label{le:cofaselBadError}
	For a given constant $c>8$ there exist constants $c_1,c_2>1$, such that the \CoFaSel Protocol as given in \cref{alg:CoFaSel} outputs a data item $d$ with a rank
	in $[log^{c_1}(n)\cdot 7k,log^{c_2}(n)\cdot 7k]$ with probability at least $1-n^{-c}$.
\end{lemma}
\begin{proof}
	We use the same simple argumentation as in Lemma~\ref{le:cofaselConst}, but instead consider a larger amount of nodes that participate in the binomial experiment. 
	Let $X$ denote the rank of the data item $d$ which is identified by \cref{alg:CoFaSel}, and let $Y$ be drawn by $Bin(n = 7k \log^c (n), p = \phi^{h_{min}-1})$.
	Observe that $\E[Y] = 7k\cdot \log^c (n) \cdot (\phi^{h_{min}-1}) = 7k\cdot \log^c (n) (7k)^{-1} \phi^{-\alpha} = \log^c (n) \phi^{-\alpha}$ holds and thus,
	$\Pr[X > 7k \log^c (n)] \leq \Pr[Y < \frac{1}{\phi^\alpha}] \leq \Pr[Y < (1-(1-\frac{1}{\log^c (n)})) \log^c (n) \phi^{-\alpha}] \leq \exp(-\frac{1}{2} \log^c (n) \phi^{-\alpha} (1-\frac{1}{\log^c (n)})^2) \leq \exp(-\frac{1}{8} \log^c (n)) \leq n^{-\frac{1}{8}c}$.
\end{proof}

\begin{lemma}
	\label{le:cofaselApproxRounds}
	Let $N>k$ and $h_{max} \geq \log_{\nicefrac{1}{\phi}} (N)$ hold. 
	The \CoFaSel Protocol presented in \cref{alg:CoFaSel} uses an amount of at most $\bigO ( \frac{1}{\phi} \ (\log_{\nicefrac{1}{\phi}} ( \frac{N}{k}) + 1) )$ messages in expectation.
\end{lemma}
\begin{proof}
	Consider one instance of \CoFaSel and applying arguments from \cref{th:topk}, the algorithm uses $\frac{1}{\phi}$ messages in expectation for each iteration of Steps~4 to 6.
	Taking $h_{max}-h_{min} + 1$ repetitions of Steps~4. to 6., and an expected amount of $\frac{1}{\phi}$ messages per repetition into account, the total number of messages follows as claimed.
\end{proof}

We apply a standard boosting technique, i.e.\ we use $ \bigO(\log (\frac{1}{\delta'}))$ independent instances of \cref{alg:CoFaSel}, and consider the median of the outputs of all instances to be the overall output.
We denote this amplified version of \CoFaSel by \emph{\CoFaSelAmp}. 
Thus, an output in the interval $[k, 42\,k]$ with probability at least $1-\delta'$ is determined.

Since we run the $\bigO( \log (\frac{1}{\delta'}))$ instances in parallel, and the server is able to process all incoming messages within the same communication round,
the number of communication rounds does not increase by this extension of the protocol.
These simple observations lead to the following theorem summarizing a first result for the $k$-select problem:

\begin{theorem}
\label{th:CoFaSel}
Let $N>k$ and $h_{max} \geq \log_{\nicefrac{1}{\phi}} (N)$ hold. 
Let $\delta'$ be a given constant.
The algorithm \CoFaSelAmp determines a data item $d$ with rank at least $k$ and at most $42 k$ with probability at least $1-\delta'$ using $\bigO ( \frac{1}{\phi} \log_{\nicefrac{1}{\phi}} (\frac{N}{k}) \, \log (\frac{1}{\delta'}) )$ messages in expectation and $h_{max}-h_{min} + 2$ communication rounds. 
\end{theorem}

\subsection{Approximate $k$-Select}
\label{se:aproxkSelect}
In this section we propose an algorithm which is based on the algorithm from the previous section. 
Here, we aim for an $(\varepsilon, \delta)$-approximation of the $k$-Selection problem for a single time step. 
Using the approximation given by \CoFaSelAmp, which gives a data item $d$ with a rank between $k$ and $42k$ with probability at least $1-\delta'$, a simple standard sampling strategy is applied afterwards.
Note that only those nodes take place in this strategy which observed a data item $d_i$ smaller than $d$.  

\begin{algorithm}[ht]
\begin{enumerate}
\item Call \CoFaSelAmp\!\!$(h_{max}, \phi, k, \delta')$ and obtain data item $d$.
\item Each node $i$ with $d_i < d$:
\item \quad Toss a coin with $p \coloneqq \min \left(1, \frac{c}{k}\mathcal{S}_{\varepsilon, \delta} \right)$. 
\item \quad On success send $d_i$ to the server. 
\item The server sorts these values and outputs $d_K$, the $p\cdot k$-th smallest item.
\end{enumerate}
\caption{\aproxkSelectProtocol ApproKSel($k, \phi, \varepsilon, \delta', \delta,h_{max}$)}
\label{alg:approxkSelect}
\vspace{-3mm}
\end{algorithm}

In the following, we show that \cref{alg:approxkSelect} is an $(\varepsilon, \delta)$-approximation of the $k$-Selection protocol in \cref{th:approxkSelect}. 
We discuss a possible choice of parameters in \cref{co:approxkSelect}.

\begin{theorem}
	\label{th:approxkSelect}
	Let $N>k$ and $h_{max} \geq \log_{\nicefrac{1}{\phi}} (N)$ hold. 
The \aproxkSelectProtocol selects data item $d_K$ with rank in  $[(1-\varepsilon)\,k, (1+\varepsilon)\, k]$ with probability at least $1-\delta$ using 
$\bigO( ( 1 + \frac{\log^c (n)}{k} \delta' ) \sample  + \frac{1}{\phi} \log_{\nicefrac{1}{\phi}} (\frac{N}{k}) \log (\frac{1}{\delta'}) )$ msg.\ in exp.\ and $h_{max}-h_{min} + 3$ comm.\ rounds.
\end{theorem}
\begin{proof}
	
	From \cref{th:CoFaSel} we get that the \CoFaSelAmp protocol uses at most $\bigO( \frac{1}{\phi} \log_{\nicefrac{1}{\phi}} (\frac{N}{k}) \log (\frac{1}{\delta'}))$ messages on expectation and runs for $h_{max} - h_{min} +2$ communication rounds.
	The remaining steps of \cref{alg:approxkSelect} need only one communication round and thus the stated bound on the communication rounds follows. 
	We omit the proof for the correctness of the algorithm, i.e.\ with demanded probability the $k$-th smallest data item is approximated, since it is based on a simple argument using Chernoff bounds. 
	
	It remains to show the upper bound on the number of messages used. 
	Formally, we apply the law of total expectation and consider the events that the $\CoFaSelAmp$ protocol determined a data item $d$ with rank $k \leq rank(d) \leq 42k$ and the event $rank(d) > 42k$.
	
	Observe that the sampling process in steps $2$ and $3$ yields $\bigO(\frac{rank(d_K)}{k} \sample)$ number of messages in expectation.
	Consider the event \CoFaSelAmp determined a data item $d$ with rank $k \leq rank(d) \leq 42 k$. 
	Then, the \aproxkSelectProtocol uses $\bigO(\sample)$ messages in expectation.
	Now consider the event \CoFaSelAmp determined a data item $d$ with $d > 42\, k$.
	We upper bound the number of messages used for this case by the rank of the given value $r$:
	\man{i think just data item $d$ ersetzen statt value $r$, aber nochmal durchlesen}
	It uses $\bigO\left(\frac{\log^c (n)}{k} \sample  \right)$ messages in expectation.
	Since the probability for this event is upper bounded by $\delta'$, the conditional expected number of messages is $\bigO\left( \frac{\log^c (n)}{k} \sample \cdot \delta' \right)$.
\end{proof}

For the sake of self containment we propose a bound which considers all nodes to take part in the protocol ($N = n$).
Note, that the \CoFaSelAmp protocol outputs a value with rank smaller than $7k \cdot \polylog (n)$ w.h.p.\ (c.f.\ \cref{le:cofaselBadError}). 
\begin{corollary}
	\label{co:approxkSelect}
Let $c$ be a sufficiently large constant.
Furthermore, let $N = n$, $\phi \coloneqq \frac{1}{2}$, $h_{max} \coloneqq \log n$, and $\delta' \coloneqq \frac{1}{\log ^c (n)}$. 
The protocol uses an amount of at most $\bigO(\sample + \log (n) \log (\log (n)))$ messages in expectation and $\log (\frac{n}{k})$ rounds of communication. 
\end{corollary}
This represents the case that a small number of messages and a large number of communication rounds are used. 
This observation is complemented by a lower bound of $\Omega(\log (n))$ in \cref{se:lowerBounds}.
Note, that this bound can be reduced to $\bigO(\sample + \log (n))$ by running one instance of $\CoFaSel$ until $h'_{min} \coloneqq \lceil \log(7k) \rceil + c$ and denote the output (i.e.\ the smallest data item) as $d$.
The \aproxkSelectProtocol is then applied only on nodes that observed data items smaller than $d$. 

\begin{corollary}
	\label{co:approxkSelect-topk}
	Let $N = n$, $\phi \coloneqq \frac{1}{2}$, $h_{max} \coloneqq \log n$, $k' \coloneqq 2 k$, $\epsilon \coloneqq \frac{1}{2}$ and $\delta' \coloneqq \frac{1}{\log ^c (n)}$. 
	The protocol uses $\bigO(k + \log (n))$ messages in expectation to solve the \topkproblem. 
\end{corollary}
\man{proof in appendix?}

\section{Multiple Time Step Computation:\\
	 A Fully Dynamic Distributed Data Structure}
\label{sec:DS}

In this section we consider computations of Top-$k$ or approx. $k$-Select for multiple time steps. 
We use a dynamic data structure to keep rough information such that a required computation can be executed more efficiently.
The main part of this section focuses on computing an element with rank close to $k$ utilizing our data structure.
The final results for answering the queries on the basis of this element are described in Section~\ref{sec:final}.

Our basic idea is to maintain a structure similar to the trees (to be more precise, only the left-most path) used to identify the (approximately) $k$'th smallest items in the previous chapters.
The data structure maintains the rough rank sketch which is defined as follows:
\begin{definition}[Rough Rank Sketch (RRS)]
	\label{p:logsel}
	A data structure for the approximate $k$-select problem fulfills the RRS property if
	a request for the data item of rank $k$
	will be answered with an item of rank in $[k,\log^c(n)\cdot k]$ with probability at least $1- \log^{-c} (n)$.
\end{definition}

We divide the ranks $1,\ldots,n$ into classes.
The goal is that a data item of each class (representative) is contained in our data structure.
The height of a class represents the expected maximum height found within this class, such that our representative will have a height value
within the noted bounds.

Let $\h \coloneqq \logphilog$.
The idea of classes is captured in the following definition:

\begin{definition}
	Let $\con$ be sufficiently large.
	A \emph{Class} $\pazocal{C}_\ell^t$ consists of all data items $d_j^t$ with $rank(d_j^t) \in [\log^{\ell 8\con} (n), \log^{(\ell+1) 8\con}(n))$.
	We denote by $h(C_\ell^t) = (\ell 8\con \h, (\ell+1) 8\con \h ]$ the height of the class $\pazocal{C}_\ell^t$.
\end{definition}
By abuse of notation we introduce $d_i^t \in C_\ell^t$ which shortens $rank(d_i^t) \in C_\ell^t$.

We divide each class into sub-classes as follows:

\begin{definition}
	Let $\con$ be defined as before. 
	We denote by a \emph{sub class} $C_{\ell, \tau}^t$, with $\tau \in \{0, \ldots, 3\}$, the set of data items $d_i^t$ with a rank $rank(d_i^t)$ between  $\log^{\ell 8\con + 2 \tau \con}(n)$ and $\log^{\ell 8\con + (2 \tau + 2) \con} (n)$. 
	The height of $C_{\ell, \tau}^t$ is $h(C_{\ell, \tau}^t) = ((\ell 8\con + (2 \tau + 1) \con) \h, (\ell 8\con + (2 \tau + 3) \con) \h]$.
\end{definition}

We omit the time step $t$ in our notation whenever it is clear from the context.

\begin{definition}
	The data items in a class $C_\ell$ are \emph{well-shaped}, if for each data item $d_i$ with	$rank(d_i) \in [\log^{\ell 8\con + 2 \tau \con}(n), \log^{\ell 8\con + (2 \tau + 2) \con}(n)]$ it holds 
	 $h_i \leq (\ell 8\con + (2 \tau + 3) \con) \h$.
\end{definition}

\begin{algorithm}[ht]

INITIALIZE() [Repeat until all classes are filled, i.e.\ $\forall \ell \ \exists \al \in S_{\ell,1}, \rl \in S_{\ell, 2}$] \squash
\begin{enumerate}[noitemsep]
	\item \textbf{Call} \CoFaSel\!$(\phi, h_{max} = \log_{1/\phi} n, k = 1)$ (\cref{alg:CoFaSel}), and keep a data structure $DS$ with all $(h_i, d_i)$ pairs, where for each $h_i$ the smallest response $d_i$ is kept.
	\item Assign data item $d_i$ with height $h_i$ to its sub class $S_{\ell, \tau'}$, if $h_i \in h(C_{\ell, \tau'})$ holds.
	\item Choose $\al \in S_{\ell, 1}$ and $\rl \in S_{\ell, 2}$ uniformly at random.
\end{enumerate}

UPDATE($i, d$) [Executed by node $i$] \squash
\begin{enumerate}[noitemsep]
	\item Update $d_i^t$ by $d_i^{t+1} = d$, delete $(d_i^t, h_i^t)$ from $DS$ (if it was in $DS$).
	\item \textbf{If} $d_i = \al$ or $d_i = \rl$ or ($h_i \in [h(\al), h(\rl)]$ and $d_i < \al$) \\
		\textbf{then} delete all $(h_j, d_j)$ pairs from $DS$, where $d_j \in S_\ell$ holds.
	\item Draw a new value from the geometric distribution with $p = (1 - \phi)$ and redefine $h_i \coloneqq \min\{h_i, h_{max}\}$.
\end{enumerate}

REFRESH() [Repeat until all classes are filled, i.e.\ $\forall \ell \ \exists \al \in S_{\ell,1}, \rl \in S_{\ell, 2}$] \squash
\begin{enumerate}[noitemsep]
	\item define $t \coloneqq t+1$
	
	\item Determine level $\ell$ such that all classes $C_{\ell'}$, $\ell' > \ell$ are filled, \\
	more formally $\forall \ell' > \ell, S_{\ell, 1} \neq \emptyset \wedge S_{\ell, 2} \neq \emptyset$.
	
	\item Determine maximal height $h$ of all nodes $i$ that observed an UPDATE since the last INITIALIZE or REFRESH operation.
	Let $\ell''$ be the level with $h \in h(C_{\ell''})$.\\
	Define $\ell \coloneqq max( \ell, \ell'')$.
	
	\item \textbf{Call} INITIALIZE() (only on sensor nodes $i$ with $d_i < r_\ell$)
\end{enumerate}
		
ROUGH-RANK(k) \hfill $\triangleright ~ k$ denotes a rank \squash
\begin{enumerate}[noitemsep]
	\item Determine $\ell$ such that $k \in C_{\ell-1}$ holds.
	\item \textbf{Output} representative $\rl  \in S_{\ell,2}$. 
\end{enumerate}

\caption{SeleMon($\phi$) \hfill [Select and Monitor]}
\label{alg:SeleMon}
\vspace*{-3mm}
\end{algorithm}

\subsection{Correctness of INITIALIZE}
\label{se:init}
We start by analyzing the outcome of the INITIALIZE operation.
In this, we show that a class is well-shaped with sufficiently large probability in \cref{le:wellshaped} and argue that the data structure yields a RR-Sketch in \cref{th:selemonINIT}, afterwards.

\begin{lemma}
	\label{le:wellshaped}
	Let $\ell\in\mathbb{N}$.
	After an execution of INITIALIZE, the class $C_\ell$ is well-shaped 
with probability at least $1- \log^{-c} (n)$, for some constant $c$.
\end{lemma}
\begin{proof}
	Fix a sub class $C_{\ell, \tau}$ and consider the data items $d_i$ with $rank(d_i) \in C_{\ell, \tau} = [\log^{\ell 8\con + 2\tau \con}(n), \log^{\ell 8\con + (2\tau + 2) \con} (n)]$. 
	The sub class admits the well-shaped property, if each data item has a height of at most $h \coloneqq (\ell 8\con + (2 \tau + 3) \con)\h$. 
	To this end, we upper bound the probability that there is a data item with a height of at least $h$ by applying union bound as follows:		
	\begin{align*}
		\Pr[\exists d_i \in C_{\ell, \tau} : h_i > h] 
		& \leq \left( \log^{\ell 8\con + (2 \tau + 2) \con} (n) - \log^{\ell 8\con + (2 \tau) \con} (n) \right) \cdot \phi^{h}\\
		& \leq  \log^{\ell 8\con + (2 \tau + 2) \con} (n)  \cdot \log^{- (\ell 8\con + (2 \tau + 3) \con)} (n) \leq \log^{- \con} (n) 
	\end{align*}
	Since there are $4$ sub classes in class $C_\ell$, the probability that there exists a data item which prevents the class to be well shaped is upper bounded by $4 \log^{-\con}(n)$ applying union bound once again.
\end{proof}

\begin{lemma}
	\label{le:existence}
	Consider a sub class $C_{\ell, \tau'}$, with $\tau' \in \{1, 2\}$.
	There is a data item $d_{i} \in C_{\ell, \tau'}$ with $h_{i} > (\ell 8\con + (2 \tau' + 1)\con) \h$ with high probability.
\end{lemma}
\begin{proof}
	Recall that for a fixed data item $d_i$ and sensor node $i$ the probability for $h_i > h$ is $\phi^h$. 
	Here we simply upper bound the probability that each data item in the sub class has a height of at most $h$ as follows:
	\begin{align*}
		\Pr[\forall d_i & \in C_{\ell, \tau'}: h_i \leq (\ell 8\con + (2\tau' + 1) \con) \h] 
		\leq \left( 1 - \phi^{(\ell 8\con + (2\tau' + 1) \con) \h}   \right)^{|C_{\ell, \tau'}|}\\
		& \leq \left( 1 - {\log^{-(\ell 8\con+ (2\tau' + 1)\con)}(n)}  \right) ^{\log^{\ell 8\con + (2\tau' + 2) \con} (n) - \log^{\ell 8\con + (2\tau')\con} (n)}\\
		& \leq \left( \frac{1}{e} \right)^{\frac{1}{2} \log^{\con}(n) }
		\leq n^{-\frac{1}{2} \log(e) \log^{\con - 1}(n)}
		\leq n^{-c},
	\end{align*}
	for some constant $c$.
\end{proof}

\begin{theorem}
	\label{th:selemonINIT}
	After execution of INITIALIZE for each rank $k$ exists a data item in the data structure with rank between $k$ and $k \cdot \log ^c (n)$ with probability at least $1-\log^{-c'} (n)$ for constants $c, c'$.
\end{theorem}
\begin{proof}
	
	First consider a fixed class $C_\ell$ for a fixed $\ell \in \mathbb{N}$.
	Based on \cref{le:wellshaped} we can show that the distribution of the random heights is well-shaped with probability at least $1-\log^{-c^{*}}(n)$ for a constant $c^{*}$.
	Now, with high probability there is a data item with such a height for sufficiently large $\con$ and $n$ due to \cref{le:existence}.
	We may in fact choose $c^{*}$ such that the probabilities for both to occur is at least $1-\log^{-c^{*}}(n)$.
	These observations together show that there is a data item $d_\tau'$ identified and stored in $DS$ and thus, for each request $k \in C_{\ell-1}$ the algorithm has identified a representative in $C_{\ell}$ as a response with a rank only by a polylogarithmic factor larger than $k$.

	Furthermore, note that there are at most $\log(n)$ number of classes.
	The argument stated above applied to each class leads to the desired result, where (applying union bound) also shows the desired success probability of $1 - \log^{-(c^{*}-1)}(n)$.
\end{proof}

Now we have shown that the Rough Rank Sketch is calculated by executing INITIALIZE with certain probability.
To analyze the number of messages in expectation and the number communication rounds we refer to \cref{le:cofaselApproxRounds} and \cref{th:CoFaSel}, respectively.
Since INITIALIZE is strongly based on the \CoFaSel protocol, similar arguments hold for this section, again.
However, note that the repetitions of the algorithm to obtain representatives $\al$ and $\rl$ for each level $\ell$ and thus for each class $C_\ell$ is not straight forward. 
A complete recomputation from scratch until all representatives are obtained introduces a factor of $\log^*(n)$ to the communication costs and rounds.
For this simply observe that there are at most $\log(n)$ different classes $C_\ell$, for which 
However, here only for those levels the \CoFaSel protocol is called, where $\al$ and $\rl$ is not known leading to additional constant factor overhead (in expectation).

\subsection{Correctness of REFRESH}
\label{se:refresh}

In the previous subsection we have shown that the algorithm INITIALIZE computes a Rough Rank Sketch for a fixed time step.
In this section we show that the REFRESH method preserves and / or rebuild parts of the data structure such that a Rough Rank Sketch is achieved after $m$ UPDATES took place. 
We analyze two different scenarios and analyze the probability of the scenario to occur: 
The representative of the class itself is UPDATED and thus, the class gets deleted
and the case that the representative does not get an UPDATE, but the rank does not reflect the situation correctly at the next time step $t+1$. 

\subsubsection{UPDATE to a representative}
\label{sec:attack}
We analyze the probability that the alarm $\al$ or the representative $\rl$ of a class $C_\ell$ is updated and thus, the class gets deleted from the data structure. 
This is okay, if there are sufficiently many UPDATEs, i.e.\ $m$ is sufficiently large.
However, if $m$ is small compared to the number of data items sub classes $C_{\ell,1}$ and $C_{\ell, 2}$ consists of, the probability to choose exactly $\al$ or $\rl$ for an UPDATE and thus delete from the data structure is small, as \cref{le:update} states.
The proof can be found in the appendix.

\begin{lemma}
	\label{le:update}
	Let $m$ be the number of UPDATE operations since INITIALIZE or REFRESH is called. 
	Let $C_\ell$ be a class with $m < \log^{\ell 8\con} (n)$. 
	The representative of $C_i$ did not get an update with probability at least $1-\log^{-c}(n)$, for a constant $c$.
\end{lemma}
\begin{proof}
	Recall that the algorithm deletes the element $d_i$ from $DS$ if UPDATE($i, d$) is called for an arbitrary $d$.
	The probability that the representative is updated is maximized if the UPDATE operations occur on $m$ different nodes.

	The size of $C_{\ell, \tau'}$ is at least $\log^{\ell 8\con + (2 \tau' + 2) \con} (n) - \log^{\ell 8\con + 2 \tau' \con}(n)$ which is larger than $\frac{1}{2} \log^{\ell 8\con + (2 \tau' + 2) \con}(n)$ and $m \leq M = \log^{\ell 8\con + 1}(n)$. 
	The probability for \cref{alg:SeleMon} to randomly choose one of these nodes as a representative is upper bounded by the term $2 \log^{ - 2 \con}(n) \leq \log^{-c}(n)$.
\end{proof}

\subsubsection{Push representative out of class} 
\label{sec:push}
We want to estimate the probability that some representative $d_i$ from our data structure was in $C_\ell^t$ but not in $C_\ell^{t+1}$.

We start our analysis with a result on the rank of a data item, given the randomly drawn height in \cref{le:position}. 
Afterwards, we show in \cref{le:updateClassSmall} that the representative of $C_\ell^t$ is still in $C_\ell^{t+1}$, if $m$ is not too large.
However, if the number of updates is large, we analyze that the representative is deleted from $DS$ in \cref{le:updateClassLarge} with sufficiently large probability.
We conclude that the desired properties of the data structure are restored after a REFRESH operation (\cref{th:refresh}).

\begin{lemma}
	\label{le:position}
	Fix a time $t$. Consider a data item $d_i \in DS$ and let $h_i \in h(C_{\ell, \tau'})$ be the height of data item $d_i$.
	It holds $d_i \in C_{\ell, \tau'}$ with high probability.
\end{lemma}
\begin{proof}
	\man{Technik aus dem 'Mergeable Summaries' paper zitieren}
	We can apply a Chernoff argument to bound the probability that the rank of a given data item is of a specific range as follows:
	Consider data item $d_i$ and fix the height $h_i$ node $i$ has drawn randomly.
	Since $d_i \in DS$ holds, $d_i$ is the smallest data item among all data items with height $h_i$.
	Intuitively speaking, the rank of the data item with height $h_i$ is simply the number of repetitions of the random experiment until there is a 'success', meaning the height $h_i$ is drawn. 
	
	We bound the probability for the events $rank(d_i) < \log^{\ell 8\con + 2 \tau' \con}(n)$ denoted by $\pazocal{E}_1$ and $rank(d_i) > \log^{\ell 8\con + (2 \tau' + 2) \con}(n)$ denoted by $\pazocal{E}_2$.
	\begin{align*}
		\Pr[\pazocal{E}_1] 
		& = \Pr[\exists i: rank(d_i) < \log^{\ell 8\con + 2 \tau' \con} (n) \wedge h_i \in h(I_\ell) ]\\
		& \leq \Pr[\exists i: rank(d_i) < \log^{\ell 8\con + 2 \tau' \con } (n) \wedge h_i > (\ell 8\con + (2\tau' + 1) \con) \h ]
	\end{align*}
	To apply a Chernoff bound consider the expected number $X$ of 'successful' coinflips, where $h_i > (\ell 8\con + (2 \tau' + 1) \con)\h$ is a successful coin flip:
	$$\gamma_X \coloneqq \E[X] 
	= \log^{\ell 8\con + 2 \tau' \con} (n) \cdot \phi^{(\ell 8\con + (2 \tau' + 1) \con) \h} 
	= \log^{-\con} (n)$$
	Now consider the probability that there is a node with a coin success and a small rank to upper bound $\Pr[\pazocal{E}_1]$ as follows:
	\begin{align*}
		\Pr[X > (1 + (\log^{\con} (n) - 1) ) \gamma_X]
		& \leq \exp(- \frac{1}{12} \cdot (\log^{\con}(n))^2 \gamma_X)
		\leq n^{-\frac{1}{12} \log^{\con-2} (n)} 
	\end{align*}
	This can be upper bounded by $n^{-c}$ for $\con \geq 3$ and sufficiently large $n$.
	
	The argument for $\Pr[\pazocal{E}_2]$ follows similar ideas:
	\begin{align*}
		\Pr[\pazocal{E}_2] 
		& = \Pr[\forall i: rank(d_i) < \log^{\ell 8\con + (2 \tau' + 2) \con} (n) \Rightarrow h_i < (\ell 8\con + (2 \tau' + 1) \con) \h ]	
	\end{align*}
	Now consider the expected number $Y$ of 'successful' coin flips, where $h_i \in h(I_\ell) $ is a successful coin flip:
	\begin{align*}
		\gamma_Y = \E[Y] 
		& \leq \log^{\ell 8\con + (2 \tau' + 2) \con} (n) \cdot \phi^{(\ell 8\con + (2 \tau' + 1) \con) \h}  
		= \log^{\con} (n) 	\end{align*}
	Now consider the probability that there is no node with a coin success:
	$$\Pr[\pazocal{E}_2] \leq \Pr[Y < (1 - \frac{1}{2} ) \gamma_Y]
	\leq \exp(- \frac{1}{12} \cdot \gamma_Y)
	\leq n^{-\frac{1}{12} \log^{\con - 2} (n)} \leq n^{-c},
	$$
	with $\con \geq 3$ and sufficiently large $n$.
\end{proof}

\begin{lemma}
	\label{le:updateClassSmall} 
 	Fix an $\ell$ with $m < \log^{\ell 8\con} (n)$. 
 	Let $t_0$ be the time step INITIALIZE was called and $t - t_0 \leq \log (n)$ hold.
 	If at every time $t' \in [t_0, t]$ it holds $m_{t'} < \log^{\ell 8\con}(n)$, then the representative of $C_\ell^{t_0}$ is a valid representative of $C_\ell^t$ (w.h.p.).
\end{lemma}
\begin{proof}
	Fix the sensor node $i$ which observed $d_i = r_\ell$. 
	Recall that $r_\ell \in C_{\ell, 2}^{t_0}$ holds.
	Now define the following: 
	Denote by $m_1$ the number of nodes $j$ that observed a data item $d_{j}^{t_0} < d_i$ at time $t_0$ and $d_j^t > d_j$ at time $t$.
	Additionally, let $M \coloneqq \log (n) \cdot \log^{\ell 8\con}(n) = \log^{\ell 8\con + 1} (n)$ be the upper bound on the total number of UPDATE operations since the last INITIALIZE operation. 
	
	We first consider the case that $m_1$ is maximal, i.e. $m_1 = M$.
	Recall that $rank(d_i^{t_0}) = rank(r_\ell) \in C_{\ell, 2}$ holds with high probability due to \cref{le:position}, i.e. $rank(d_i^{t_0}) \geq \log^{\ell 8\con + 2 \con} (n)$ holds.
	This leads to the simple observation that $rank(d_i^t) \geq \log^{\ell 8\con + 2 \con}(n) - M \geq \log^{\ell 8\con}(n)$ holds with high probability.
	
	The argument for a maximal $m_2$ is analog, the claim follows.
\end{proof}

\begin{lemma}
	\label{le:updateClassLarge}
	Fix an $\ell$ with $m \geq \log^{\ell 8\con} (n)$. 
	Let $t_0$ be the time step INITIALIZE was called and $t - t_0 \leq \log (n)$ hold.
	If 
	$\rl \in C_{\ell, 2}^t$ is no longer in $C_\ell^{t+1}$ the protocol deletes $\rl$ with high probability.
\end{lemma}
\begin{proof}
	Fix $\al \in C_{\ell, 1}^t$ and $\rl \in C_{\ell, 2}^t$. 
	Consider two cases: (1) there are (up to) $m$ UPDATES such that the rank of $\rl$ increases and (2) (up to) $m$ UPDATES such that the rank of $\rl$ decreases. \newline
	(1) Recall that $\rl \in C_{\ell,2}^t$ holds, where $d_i\in C_{\ell, 2}^t$ if for the rank of data item $d_i$ it holds $rank(d_i) \in [\log^{\ell 8\con + 4 \con} (n), \log^{\ell 8\con + 6 \con}(n)]$.
	The ranks of the data items in class $C_\ell$ are upper bounded by $\log^{\ell 8\con + 8 \con}(n)$.
	Thus, there are at least a number of $\log^{\ell 8\con + 8 \con}(n) - \log^{\ell 8\con + 6 \con}(n) \gg \frac{1}{2} \log^{\ell 8\con + 8 \con}(n)$ data items which observe an UPDATE and flip a coin due to the UPDATE algorithm. 
	By the same argument as in \cref{le:wellshaped} and \cref{le:existence}, there is one data item $d_{new}$ with a height strictly larger than the height of the data item $\al$ (with high probability).
	Thus, the data items $\al$ and $\rl$ get deleted from $DS$ with high probability.\newline
	(2) For this case we argue that it is unlikely that the data item $\rl$ which was a representative at time $t$ has a rank smaller than $\log^{\ell 8\con}(n)$ at time $t+1$ without observing an update of data item $\al$:
	Recall, that $\rl \in C_{\ell,2}^t$, i.e.\ $rank(\rl) \geq \log^{\ell 8\con + 4 \con}(n)$. 
	There have to be at least $\log^{\ell 8\con+ 4 \con}(n) - \log^{\ell 8\con + 2 \con} (n)$ UPDATE calls such that $\rl$ is incorrectly in $C_\ell^{t+1}$ (if it is not deleted).
	In case $\al \in C_{\ell,1}$ holds, the respective sensor node observes an UPDATE to the data item $\al$ with high probability followed by a deletion of $\rl$, concluding the proof.	
\end{proof}

\begin{theorem}
	\label{th:refresh}
	After each REFRESH operation the Rough Rank Sketch is restored, i.e.\ for each rank $k$ there exists a data item $d$ in the data structure $DS$ as defined in \cref{alg:SeleMon} with rank between $k$ and $k \cdot \log ^c (n)$ with probability at least $1-\log^{-c'}(n)$.
\end{theorem}
\begin{proof}
	First, observe that this property holds after the INITIALIZE operation due to \cref{th:selemonINIT}.
	It remains to show that this property is preserved or refreshed after an UPDATE operation.
	We argue for an arbitrary but fixed class $C_\ell$.
	Consider the following case distinction: (1) The representative $\rl$ of $C_\ell$ is updated and thus, the class gets completely rebuilt and (2) the representative $\rl$ of $C_\ell$ is not updated. \newline
	\textit{[Representant $\rl \in C_\ell$ is updated]} 
	In this case, each data item $d \in DS$ which is assigned to $C_\ell$ gets deleted.
	The structure gets rebuilt in the REFRESH operation.
	Thus, the correctness of RSS follows by similar arguments as the INITIALIZE operation. \newline
	\textit{[Representant $\rl \in C_\ell$ is not updated]}
	Now consider the (sub-)cases whether $m < \log^{\ell 8\con + 4 \con}(n)$ holds.
	In case it holds, the representative of $C_\ell^t$ is still valid at time $t+1$, i.e.\ $\rl \in C_\ell^t \cap C_\ell^{t+1}$ with high probability due to \cref{le:position} and \cref{le:updateClassSmall}.
	On the other hand, if $m$ is large, then $\rl$ gets deleted from $DS$ (and hence the whole class $C_\ell$) due to \cref{le:updateClassLarge}. 
\end{proof}

Finally, we show the number of messages the data structure uses in order to build or rebuild the Rough Rank Sketch.

\begin{theorem}
\label{th:messages}
	The operations INITIALIZE and REFRESH use $\bigO(\frac{1}{\phi} \log_{\nicefrac{1}{\phi}} (n))$ and $\bigO(\frac{1}{\phi} \log_{\nicefrac{1}{\phi}} (m))$ messages in expectation, respectively.
\end{theorem}
\begin{proof}
	First, note that the bound on the communication used by INITIALIZE follows by the same arguments as \cref{le:cofaselApproxRounds}, since INITIALIZE only calls the \CoFaSel protocol. 
	Since there is no early termination rule, i.e.\ $h_{min} = 1$ holds, the protocol uses an amount of $\bigO(\frac{1}{\phi} \log_{\nicefrac{1}{\phi}} (n))$ messages in expectation as claimed above. 
	
	Second, for the upper bound of REFRESH, we argue using the law of total expectation and consider the smallest $\ell$ such that all $\ell' > \ell$ do not need to be refreshed, i.e.\ all representatives of $S_{\ell'}$ remain in the data structure.
	Now consider the event that $m < \log^{\ell 8 \con}(n)$ holds. 
	It directly follows that the number of messages to compute the representatives of classes $C_{\ell''}$ with $\ell'' \leq \ell$ needs $\bigO(\frac{1}{\phi} \log_{\nicefrac{1}{\phi}}(m)$ number of messages as claimed. 
	However, if $m \geq \log^{\ell 8 \con}(n)$ holds, observe that this event happens with a probability $\log^{-c}(n)$, where $c$ directly depends on the choice of $\con$. 
	For $c > 1$ the conditional expected costs yield costs of only $\bigO(\frac{1}{\phi})$, such that the overall costs follow as claimed.		
\end{proof}

\subsection{Implications for Top-$k$ and Approximate $k$-Select}
\label{sec:final}

Here we shortly describe how the data structure can be used to efficiently answer a \topk or \kselect request. 
Note that the $k$ does not need to be known beforehand, each request can be posed with a different parameter.

Both computations start with a REFRESH operation if there were UPDATEs since the last REFRESH.
We then obtain an item with rank close to $k$ with a ROUGH-RANK$(k)$ operation.

For determining the $k$ smallest items, the response from the data structure, denoted by $d$, is broadcasted such that all sensor nodes with a larger data item do not take place in the call of \topkprotocol with parameter $k$.
If this call was not successful, a second call of \topkprotocol is executed on all sensor nodes.
Considering the expected costs
conditioned on whether the first or second call was successful and applying the law of total expectation, the simple bound on the communication follows:
\begin{corollary}
One computation of the Top-$k$ needs $k + \bigO(\log(m)+ \log (\log (n)))$ messages in expectation assuming $m$ UPDATEs are processed since the last Top-$k$ query.
\end{corollary}

Exactly the same approach is used to solve the \kselect problem. 
We define the (internal) failure probability $\delta' \coloneqq \log^{-1}(n)$ and obtain the following simple bound by applying the same arguments as for the \topkprotocol.
\begin{corollary}
One computation of \aproxKSelectProblem uses $\bigO(\sample + \log(m) + \log^2(\log (n)))$ msg.\ in expectation assuming $m$ UPDATES are processed since the last query.
\end{corollary}

\section{Lower Bounds}
\label{se:lowerBounds}

In this section we consider lower bounds for the problems considered in the past sections. 
We show that our main results in the previous sections are asymptotically tight up to additive costs of $\bigO(\log(\log(n)))$ per time step for a constant choice of $\phi$, and $\bigO(\log^2(\log(n)))$, respectively. 
For scenarios in which the adversary changes a polylogarithmic number of values in each time step, the proposed bounds are asymptotically tight.

\begin{lemma}[\cite{mmm15}]
	\label{le:lowerBound}
	Every algorithm needs $\Omega(\log (n))$ messages in expectation to output the maximum in our setting.
\end{lemma}

We extend this lemma to multiple time steps and to monte carlo algorithms which solve the problem at each time step with a fixed probability:

\begin{lemma}
\label{le:secondlower}
	Every algorithm that outputs for a given $c>1$, a data item $d$ where $rank(d) \in \bigO(\log^c (n))$ holds and success probability at least $1 - \frac{1}{\log (n)}$ uses $\Omega(\log (n))$ messages in expectation in our setting.
\end{lemma}
\begin{proof}
	Assume there is an algorithm $\ALG$ which achieves the above requirements with $o(\log (n))$ messages in expectation.
	We construct an algorithm $\pazocal{A}$ which first applies $\ALG$ and then the \topkprotocol for $k \coloneqq 1$, and $h_{max} \coloneqq \log (n)$ afterwards.
	
	Observe that $\pazocal{A}$ uses $\bigO(\log (\log (n)))$ messages in expectation if the event $m \leq \log^c (n)$ occurs and $\log (n)$ messages, else. 
	The probability that the latter event occurs is upper bounded by $\frac{1}{\log (n)}$ and thus, $\pazocal{A}$ uses $o(\log (n))$ messages in expectation which contradicts \cref{le:lowerBound}.
\end{proof}

We further extend the lower bounds to multiple time steps in which an adversary is allowed to change values of at most $m$ nodes between two consecutive time steps.
It is easy to see if the instance always changes the smallest $m$ nodes between each time steps and chooses random permutations, the following holds:
\begin{theorem}
\label{th:lower}
	Every algorithm $\pazocal{A}$ that tracks the minimum over $T$ time steps uses an amount of $\Omega(\log (n) + T \ \log (m))$ messages in expectation, if m UPDATES per step are processed in our setting.
\end{theorem}
\begin{proof}
	For the proof we assume, that $T$ is at most $\log (n)$, split $T$ in $T_1, T_2, \ldots$, each of size $\log (n)$ (except for the last one).

	Now, construct an instance as follows:
	Initially define sets of data items $\pazocal{S}_0, \pazocal{S}_1, \ldots, \pazocal{S}_{\log (n)}$.
	In each time step the data items for an UPDATE are chosen from these sets.
	For each consecutive set $\pazocal{S}_t, \pazocal{S}_{t+1}$ holds that the largest data item in $\pazocal{S}_{t+1}$ is smaller than the smallest data item in $\pazocal{S}_t$. 
	Furthermore, the size of each $\pazocal{S}_t$ is $m$, for $t \geq 1$ and $n$ for $t = 0$. 
	The adversary chooses a random permutation $\pi_0$ which defines which node initially observes a data item from $\pazocal{S}_0$.
	Denote the set of nodes with the $m$ smallest data items given by the random permutation by $\pazocal{N}$.
	For each consecutive time step $t > 0$ the adversary chooses a random permutation $\pi_t$ of $\pazocal{S}_t$ and chooses the $m$ nodes in $\pazocal{N}$ to process these UPDATES. 
	
	Observe that in the first step $(t = 0)$ based on \cref{le:lowerBound} we argue that at least $\Omega(\log (n))$ messages are used to identify the minimum. 
	For each consecutive time step $(t > 0)$ the adversary chooses the same set of nodes $\pazocal{N}$ to process UPDATES. 
	Based on the construction of $\pazocal{S}_{t}$ and $\pazocal{S}_{t-1}$ no information of $\pazocal{S}_{t-1}$ can be exploited by any algorithm to proceed the permutation of $\pazocal{S}_t$ using less than $\Omega(\log (m))$ messages in expectation, concluding the proof.
\end{proof}

\end{document}